\documentclass[11pt]{article}
\usepackage{amsmath}
\usepackage{comment}
\usepackage{mdframed}
\usepackage{enumitem}
\usepackage{amsthm}
\usepackage{amsmath,amssymb}
\usepackage{flexisym}
\usepackage{breqn}
\usepackage{bbm}
\usepackage{url}
\usepackage{hyperref}
\usepackage[utf8]{inputenc}
\usepackage[colorinlistoftodos]{todonotes}
\usepackage{xspace}
\usepackage{url}
\usepackage[capitalise]{cleveref}
\usepackage{algorithm}
\usepackage{algpseudocode}
\usepackage[legalpaper, margin=1in]{geometry}
\usepackage{comment}
\usepackage{todonotes}
\usepackage[backend=bibtex,style=alphabetic]{biblatex}

\theoremstyle{Definition}
\newtheorem{theorem}{Theorem}
\newtheorem{definition}{Definition}
\newtheorem{lemma}{Lemma}

\newtheorem{claim}{Claim}
\newtheorem{corollary}{Corollary}
\newtheorem{remark}{Remark}
\usepackage{geometry}
\usepackage[short]{optidef}
\usepackage{float}
\newfloat{algorithm}{thp}{loa}[chapter]
\floatname{algorithm}{Algorithm}
\newif\ifcomments
\commentsfalse
\ifcomments
	\newcommand{\ben}[1]{{\textcolor{red}{[\bf Ben: #1]}}}
	\newcommand{\benedikt}[1]{{\textcolor{blue}{[\bf Benedikt: #1]}}}
	
			\newcommand{\ignore}[1]{}

\else
	\newcommand{\ben}[1]{}
	\newcommand{\benedikt}[1]{}
		\newcommand{\ignore}[1]{}

\fi

\newcommand{\val}{\mathsf{val}}
\newcommand{\weight}{\mathsf{weight}}
\newcommand{\margval}{\Delta\val}
\newcommand{\margweight}{\Delta\weight}
\newcommand{\margdens}{\mathsf{density}\Delta}
\newcommand{\densbound}{\mathsf{densityb}}

\newcommand{\dens}{\mathsf{density}}

\newcommand{\cont}{\mathsf{cont}}

\newcommand{\RR}{\mathbb{R}}
\newcommand{\ZZ}{\mathbb{Z}}
\newcommand{\PP}{\mathbb{P}}
\newcommand{\NN}{\mathbb{N}}
\newcommand{\vecb}[1]{\mathbf{#1}}
\newcommand{\xvec}{\vecb{x}}

\title{Schwartz-Zippel for multilinear polynomials mod N}

\author{Benedikt B\"unz\footnote{benedikt@cs.stanford.edu}, Ben Fisch\footnote{benafisch@gmail.com}\\ Stanford University\\
}

\begin{document}

	\maketitle
	\abstract{
	We derive a tight upper bound on the probability over $\mathbf{x}=(x_1,\dots,x_\mu) \in \ZZ^\mu$ uniformly distributed in $ [0,m)^\mu$ that $f(\mathbf{x}) = 0 \bmod N$ for any $\mu$-linear polynomial $f \in \ZZ[X_1,\dots,X_\mu]$ co-prime to $N$. We show that for $N=p_1^{r_1},...,p_\ell^{r_\ell}$ this probability is bounded by 
	$\frac{\mu}{m} + \prod_{i=1}^\ell  I_{\frac{1}{p_i}}(r_i,\mu)$
	where $I$ is the regularized beta function. 
 Furthermore, we provide an inverse result that for any target parameter $\lambda$ bounds the minimum size of $N$ for which the probability that $f(\mathbf{x}) \equiv 0 \bmod N$ is at most $2^{-\lambda} + \frac{\mu}{m}$. For $\mu =1$ this is simply $N \geq 2^\lambda$. 
 For $\mu \geq 2$, $\log_2(N) \geq 8 \mu^{2}+ \log_2(2 \mu)\cdot \lambda$ the probability that $f(\xvec) \equiv 0 \bmod N$ is bounded by $2^{-\lambda} +\frac{\mu}{m}$. 
 We also present a computational method that derives tighter bounds for specific values of $\mu$ and $\lambda$. For example, our analysis shows that for $\mu=20$, $\lambda = 120$ (values typical in cryptography applications), and $\log_2(N)\geq 416$ the probability is bounded by $ 2^{-120}+\frac{20}{m}$. We provide a table of computational bounds for a large set of $\mu$ and $\lambda$ values.
}

\section{Introduction}
The famous DeMillo-Lipton-Schwartz–Zippel (DLSZ) lemma\cite{demillo1977probabilistic,zippel1979probabilistic,schwartz1980fast}\footnote{See this blog post for a detailed history of the lemma: \url{https://rjlipton.wpcomstaging.com/2009/11/30/the-curious-history-of-the-schwartz-zippel-lemma/}} states that for any field $\mathbb{F}$, non-empty finite subset $S \subseteq \mathbb{F}$, and non-zero $\mu$-variate polynomial $f$ over $\mathbb{F}$ of total degree $d$, the number of zeros of $f$ contained in $S^\mu$ is bounded by $d \cdot |S|^{\mu-1}$ (or equivalently, the probability that $f(\mathbf{x}) = 0$ for $\mathbf{x}$ sampled uniformly from $S^n$ is bounded by $\frac{d}{|S|}$). For $\mu = 1$ this simply follows from the Fundamental Theorem of Algebra, but for multivariate polynomials, the number of zeros over the whole field could be unbounded. The computational significance of this lemma is that sampling an element from $S$ only takes $n \cdot \log_2(|S|)$ random bits but the probability of randomly sampling a zero of $f$ from $S^n$ is exponentially small in $|S|$. One of its original motivations was an efficient randomized algorithm for polynomial identity testing, but it has since found widespread application in computer science. Polynomial identity testing is the canonical randomized algorithm and is famously hard to derandomize\cite{DBLP:journals/cc/KabanetsI04}.
 
%That is, in a field the probability that a non-zero polynomial is zero at a random point chosen from any set is less than the total degree divided by the size of the set. This holds true, despite the fact that the number of roots of the polynomial is exponential in the number of variables of the polynomial. 
The classical lemma applies more broadly to integral domains, but not to arbitrary commutative rings. As a simple counterexample, over the ring of integers modulo $N = 2p$ the polynomial $f(X) = pX \bmod N$ vanishes on half of the points in $[0,N)$. The lemma has been extended to commutative rings by restricting the set $S$ to special subsets in which the difference of any two elements is not a zero divisor~\cite{BishnoiCPS}. For the case of $\ZZ_N$, this means that the difference of any two elements in $S$ must be co-prime to $N$. Our present work explores the setting where $S$ is the contiguous interval $[0, m)$ and thus does not satisfy this special condition. Instead, we will restrict the polynomial $f$ to be co-prime with $N$, thus ruling out polynomials of the form $f(X) = u \cdot g(X)$ where $u$ is a zero-divisor as in the counterexample above. 

As a warmup, it is easy to see that any univariate linear polynomial $f(X)=c\cdot X+b$ co-prime to $N$ has at most one root modulo $N$. If there were two such roots $x_1 \not\equiv x_2 \bmod N$ then $c (x_1 - x_2) \equiv 0 \mod N$ implies $c$ is a zero divisor (i.e., $gcd(c, N) \neq 1$). Furthermore, $c \cdot x_1 \equiv -b \bmod N$ implies $-b = c \cdot x_1 + q \cdot  N$ for some $q \in \ZZ$, and thus, $gcd(c, N)$ also divides $b$. This would contradict the co-primality of $f$ and $N$. So for $x$ uniformly distributed in $ S= [0,m)$ the probability of $f(x)\equiv 0 \bmod N$ in this case is indeed at most $\frac{1}{|S|}$. Unfortunately, this does not appear to generalize nicely to polynomials of arbitrary degree. As an example, for $N = 2^\lambda$ the polynomial $f(X) = X^\lambda \bmod N$ vanishes on half of the points in $[0, N)$. 

On the other hand, we are able to generalize the lemma in a meaningful way to multivariate \emph{linear} polynomials (i.e, at most degree 1 in each variable). It turns out that the probability of sampling a zero from $S^\mu = [0,m)^\mu$ of a $\mu$-linear polynomial can be tightly bounded by $\frac{\mu}{|S|} + \epsilon$, where the error term $\epsilon$ is bounded by a product of regularized beta functions. 
We also formulate an inverse Schwartz-Zippel Lemma for composite showing that for all sufficiently large $N$ this error is negligibly small. In particular, for $\log N \geq 8\mu^2 + (1 + \log \mu) \lambda$ the error term is at most $2^{-\lambda}$, showing that the error decays exponentially. Our technique for deriving this threshold lower bound $t(\lambda, \mu)$ on $N$ for a target $\lambda$ formulates $t(\lambda, \mu)$ as the objective function of a knapsack problem. We derive an analytical solution by deriving bounds on the regularized beta function. We also apply a knapsack approximation algorithm to find tighter values of $t(\lambda, \mu)$ for specific values of $\mu$ and $\lambda$. 
 
\subsection{Applications and Open Problems}
We remark on one application of the Multi-Linear Composite Schwartz-Zippel lemma (LCSZ) and its inverse to a problem in cryptography, and postulate other potential cryptographic applications as open problems. Consider a rational multi-linear polynomial $g(X)=\frac{f(X)}{N}$ for $f\in \ZZ[X]$ and $N \in \NN$ coprime with $f$. The LCSZ gives a bound on the probability that $f(\xvec) \in \ZZ$ for a random point $\xvec$ as in this case $f(\xvec)\equiv 0 \bmod N$. Furthermore we can use the inverse lemma to show that if $f(\xvec)\in \ZZ$ for random point $\xvec$ then the probability that $N>B$ can be bounded.
This application of the inverse LCSZ is an essential component of a new\footnote{There was a gap in the original analysis~\cite{BlockHRRS21}.} security proof for the polynomial commitment DARK\cite{bunz2020transparent}. Additionally, it is conceivable that the inverse LCSZ could be used to strengthen the recent result called Dew\cite{ArunGLMS22}, a zero-knowledge proof system with constant proof size and efficient verifier. In particular, the current security proof requires parameters that result in quadratic prover time. There is hope that the security proof can be made to work with tighter parameters and a quasi-linear prover, using the inverse LCSZ.   

Further, the LCSZ could potentially also be used to improve other zero-knowledge proof systems for modular arithmetic over rings instead of prime fields. This can have benefits, such as enabling the use of more machine-friendly moduli (e.g., powers-of-two). For example, recent work \cite{AttemaCCDE22} generalized Bulletproofs and Compressed $\Sigma$-Protocols to work for cryptographic commitments to vectors over $\mathbb{Z}_n$ for general $n$ (not strictly prime). Their analysis also deals with multilinear polynomials over rings and uses the generalization of DLSZ~\cite{BishnoiCPS} that restricts the sampling domain $S$ to \emph{exceptional sets} (sets whose element differences are not zero-divisors). Since $\mathbb{Z}_n$ does not have large enough exceptional sets, their protocol works over finite extension rings resulting in significantly worse performance in both proof size and computation time than the original versions of these protocols over $\mathbb{Z}_p$ (e.g., by a factor $\lambda \approx 128$ for $128$-bit security). Generalizations of Bulletproofs and Compressed $\Sigma$-Protocols for lattice-based commitments (e.g., based on the hardness of the Integer SIS problem) encounter similar challenges~\cite{cryptoeprint:2021:307}. The GKR sumcheck protocol \cite{GoldwasserRK17} also partially relies on the DLSZ for multilinear polynomials. Sumcheck based protocols have similarly been extended to work over rings instead of prime fields using DLSZ with exceptional sets~\cite{chenverifiable,BootleCS21}. Our new variant (LCSZ) could possibly\footnote{A key challenge for applying LCSZ to these applications would be showing that the polynomials are co-prime with the modulus.} be used to improve these results (by eliminating the need for exceptional sets) as well as other GKR instantiations\cite{CormodeMT12,Setty20,SettyL20} to integer rings. 
%\benedikt{add something about lattices.}
%Additionally, lattice-based protocols could possibly use the composite schwartz-zippel lemma for tighter analysis. 

\subsection{Technical Overview}
The regular DLSZ is relatively simple to prove. Consider the special case of a multilinear polynomial over a field. As a base case, a univariate linear polynomial has at most one root over the field. For the induction step, express $f(X_1,...,X_{n+1}) =  g(X_1,...,X_n) + X_{n+1} h(X_1,...,X_n)$ for random variables $X_1,\dots,X_n$. The probability that $h(x_1,...,x_n) = 0$ over random $x_i$ sampled from $S$ is at most $n/|S|$ by the inductive hypothesis, and if $h(x_1,...,x_n) = w \neq 0$ and $g(x_1,\dots,x_n)=u$, then $ u + X_{n+1} w$ has at most one root (base case). By the union bound, the overall probability is at most $n/|S| + 1/|S| = (n+1)/|S|$. 

This simple proof does not work for multilinear polynomials modulo a composite integer. The base case is the same for $f$ coprime to N, which has at most one root. However, in the induction step, it isn't enough that $h(x_1,...,x_n) \neq 0$ as it still may be a zero divisor, in which case the polynomial $ u + X_{n+1} w$ is not necessarily coprime to $N$ and the base case no longer applies. The number of roots depends on $\gcd( u +X_{n+1} w, N)$. 

We prove the LCSZ using a modified inductive proof over $n$. Our analysis takes into account the distribution of $\gcd( u +X_{n+1} w,N)$. For each prime divisor $p_i$ of $N$, the highest power of $p_i$ that divides $ u + X_{n+1} w$ follows a geometric distribution. 
%In particular let $t_n$ be the highest power of $p_i$ that divides $\gcd(u,w)$. Then the probability over $X_{n+1}$ that $p^{t_{n}+\Delta}| X_{n+1} u +w$ is bounded by $\frac{1}{p_i^{\Delta}}=P[\textsf{Geo}_{p_i} \geq \Delta]$, where $\textsf{Geo}_{p_i}$ is a random geometric variable with parameter $p_i$. This probability is independent of the value of $t_n$. 
Using an inductive analysis, we are able to show that the probability $f(x_1,\dots,x_n)\equiv 0 \bmod p^r$ is bounded by the probability that $\sum_{i=1}^n Z_i \geq r$ for $i.i.d.$ geometric variables with success parameter $1-\frac{1}{p}$. This probability is equal to a $I_{\frac{1}{p}}(n,r)$ where $I$ is the regularized beta function. Furthermore, by the Chinese Remainder Theorem this probability is independent for each prime factor of $N$, and thus, the overall probability can be bounded by a product of regularized beta functions. 

\paragraph{Inverse LCSZ}
While the LCSZ gives a tight bound on the probability for particular values of $N, \mu$, and $m$, cryptographic applications require finding concrete parameters such that the probability is exponentially small in a security parameter $\lambda$. 
This is easy for the standard DLSZ, as it simply states that when sampling from sets of size $d \cdot 2^{\lambda}$ the probability is at most $\frac{d}{2^\lambda}$. The LCSZ does not have such a simple form so we need to explicitly find an inverse formulation. Concretely, we want to find a value $N^*$ such that for all $N\geq N^*$ the probability that $f(X_1,\dots,X_n)\equiv 0 \bmod N$ is bounded by $2^{-\lambda}$. 

To do this we first derive simple and useful bounds for the regularized beta function. 

\paragraph{Bounds on regularized beta function}
We show the following three useful facts about $I_{\frac{1}{p}}(n,r)$
\begin{itemize}
	\item $I_{\frac{1}{p}}(r,\mu)\leq \left(\frac{n}{p}\right)^r$ for $p \geq 2\mu$
	\item $I_{\frac{1}{p}}(r,\mu)\leq \frac{r^n}{p^r}$ for $r\geq 2\mu$
	\item $\log(I_{1/p}(r-1,\mu))-\log(I_{1/p}(r,\mu))$ is non-increasing in $r$ for any $p>\mu$ and for $r=1$ in $p$.
\end{itemize}

\paragraph{Knapsack Formulation}
We then formulate finding $N^*$ as an optimization problem. $N^*$ is the maximum value of $N$ such that the probability of $f(\xvec)\equiv 0 \bmod N$ is greater than $2^{-\lambda}$. For any $N$ let $S(N)$ denote the set of pairs $(p, r)$ where $p$ is a prime divisor of $N$ with multiplicity $r$. Taking the logarithm of both the objective and the constraint yields a knapsack-like constraint maximization  [roblem where the objective is $\log(N^*)=\sum_{(p_i,r_i) \in S(N^*)} r_i \cdot \log(p_i)$ and the constraint is $\sum_{(p_i,r_i) \in S(N^*)} -\log(I_{\frac{1}{p_i}}(r,\mu))\leq \lambda$.
Using the bounds on $I_{\frac{1}{p_i}}$ and several transformations of the problem we show that any optimal solution to this problem must be bounded by $t=8\mu^2 + \log_2(2\mu)\lambda$, which in turn implies that $N^*\leq 2^t$. 

\paragraph{Tighter computational solution}
We further show that a simple greedy knapsack algorithm computes an upper bound to the knapsack problem. The algorithm uses the fact that $\frac{\log(p)}{\log(I_{1/p}(r-1,\mu))-\log(I_{1/p}(r,\mu))}$ the so-called marginal density of each item is non increasing over certain regions. Adding the densest items to the knapsack computes an upper bound to the objective. We run the algorithm on a large number of values for $\mu$ and $\lambda$ and report the result.

%\subsection{Related Work}

\section{Main theorem statement}
\begin{theorem}[Multilinear Composite Schwartz-Zippel (LCSZ)]
\label{thm:csz} Let $N=\prod_{i=1}^\ell p_i^{r_i}$ for distinct primes $p_1,...,p_\ell$.
	Let $f$ be any $\mu$-linear integer polynomial co-prime to $N$. For any integer $m > 1$ and $\xvec$ sampled uniformly from $[0,m)^\mu$: $$\mathbb{P}_{\xvec \gets [0,m)^\mu} [f(\xvec)\equiv 0 \bmod N]\leq \frac{\mu}{m} + \prod_{i=1}^\ell 
	I_{\frac{1}{p_i}}(r_i, \mu)$$
	 where $I_{\frac{1}{p}}(r, \mu)= (1-\frac{1}{p})^\mu \sum_{j = r}^\infty \binom{\mu+r -1}{r} \left(\frac{1}{p}\right)^j$ is the regularized beta function.
\end{theorem}

%such that $$P[\sum_{j=1}^\mu Y_{p,j} \geq r] = I_{\frac{1}{p}}(r, \mu).
% $X_{p,1},\dots,X_{p,\mu}$ are independent geometric random variables with parameter $\frac{1}{p}$ such that $P[Z_{p,j} \geq r]=\left(\frac{1}{p}\right)^r$ and  $I_{\frac{1}{p}}(r, \mu)= (1-\frac{1}{p})^\mu \sum_{j = r}^\infty \binom{\mu+r -1}{r} \left(\frac{1}{p}\right)^j$ is the regularized beta function.
%\benedikt{More consistency on Negative Binomial/Regularized Beta}

\begin{remark}
The regularized beta function characterizes the tail distribution of the sum of independent geometric random variables. If $Y = \sum_{i=1}^\mu Z_i$ where each $Z_i$ is an independent geometric random variable with parameter $\epsilon$ then $P[Y \geq r]=I_{1-\epsilon}(r,\mu)$. $Y$ is a negative binomial variable with parameters $\epsilon,\mu$.
\end{remark}
\begin{remark}
For $m \gg \mu$ the theorem is nearly tight for all $N$. Setting $f(\xvec)= \prod_{i=1}^\mu x_i$ and $m = N$ gives $P_{\xvec \gets [0,m)^\mu}[f(\xvec)\equiv 0 \bmod N]= \mathbb{P}_{\xvec\gets [0,N)^\mu}[f(\xvec)\equiv 0 \bmod N] = \prod_{i=1}^\ell I_{\frac{1}{p_i}}(r_i,\mu)$ 
\end{remark}
\begin{remark}$1 - e^{-\mu/p_i} \leq I_{\frac{1}{p_i}}(1,\mu) = 1 - (1 - \frac{1}{p_i})^\mu \leq \frac{\mu}{p_i}$. Hence, for square-free $N$ the probability in \cref{thm:csz} is upper bounded by $\frac{\mu}{m} + \frac{\mu^\ell}{N}$, but for $\ell > 1$ this is a loose upper bound unless $\mu \ll p_i$ for all $p_i |N$. For $\ell = 1$  (i.e., prime $N$), \cref{thm:csz} coincides with the Schwartz-Zippel lemma.
\end{remark}
\begin{remark}
 $I_{\frac{1}{p_i}}(r_i,1)=\left(\frac{1}{p_i}\right)^{r_i}$. Hence, for $\mu=1$,  the bound in \cref{thm:csz} is $\frac{1}{N}+\frac{\mu}{m}$.
\end{remark}
%\benedikt{This makes the bound seem kind of loose no?}
%\benedikt{Replace negative binomial with regularized beta function}

\begin{proof}
We begin by introducing some notations. 

\begin{itemize}
\item For a polynomial $f \in \ZZ[X_1,...,X_\mu]$ let $\vec{f}$ denote the coefficients of $f$ and let $\cont(f)$ denote the greatest integer divisor of $f$, i.e. the \emph{content}. 
\item $\vec{\beta} = (\beta_1,...,\beta_\mu) \in [0,m)^\mu$ is a random variable distributed uniformly over $[0,m)^\mu$. For any $i \geq 1$, let $\vec{\beta}_{i} = (\beta_1,...,\beta_i)$, let $f_0 = f$,  and let $f_i(\vec{\beta}_i):=f(\beta_1,\dots,\beta_{i},X_{i+1},\dots,X_\mu)$. 
\item Given the random variable $\vec{\beta}$ distributed uniformly over $[0,m)^\mu$, define for each $j \in [1,\ell]$ and $i \in [1,\mu]$ the random variable $Y_{j,i}$ (as a function of $\vec{\beta}$) representing the multiplicity of $p_j$ in $\cont(f_i(\vec{\beta}_i)))$. Naturally, we set $Y_{j,0} = 0$ for all $j$ because $\forall_j f_0 \neq 0 \bmod p_j$. (\textbf{Note:} $Y_{j,i}$ are \emph{not} independent). For any $i \in [\mu]$ the event that $\forall_j Y_{j,i} \geq r_j$ is equivalent to the event that $f_i(\vec{\beta}_i) = 0 \bmod N$ and the event that $\forall_j Y_{j,i} = r_i$ is equivalent to $\cont(f_i(\vec{\beta}_i)) = N$. The event $\forall_j Y_{j,\mu} \geq r_j$ is thus equivalent to $f(\vec{\beta}) = 0 \bmod N$. 
\item Let $\{Z_{j,i}\}$ for $i \in [\mu]$ and $j \in [\ell]$ be a set of independent random variables, where $Z_{j,i}$ is geometric with parameter $1-\frac{1}{p_j}$.

\end{itemize}

 From the CCDF (complementary CDF) of geometric random variables we have that $P[Z_{j,i} \geq k] = (\frac{1}{p_j})^{k}$. Setting $Z_j := \sum_{i=1}^\mu Z_{j,i}$, from the CCDF of the negative binomial distribution (i.e., tail distribution of the sum of independent geometric random variables) it follows that $\forall_j P[Z_j \geq r] = I_{\frac{1}{p_j}}(r, \mu)$. 
 
 Next, we establish an important subclaim: 
 
 \begin{claim} For any $i \geq 2$ and $\vec{k},\vec{k}'\in \NN^{\ell}$ where $\forall_j k_j \geq 0$: 
 $$P[\forall_j Y_{j,i}\geq k_j + k'_j | \forall_j Y_{j,i-1} = k'_j ]\leq \frac{1}{m} + P[\forall_j Z_{j,i}\geq k_j]$$
 Furthermore,  for all $i \geq 1$, $P[\forall_j Y_{j,i}  \geq Y_{j,i-1} + k_j] \leq \frac{1}{m} + P[\forall_j Z_{j,i}\geq k_j]$.
 \end{claim}

\begin{proof}  Since the order of the variables does not matter, w.l.o.g. assume that $\forall_{j \in [1,\ell']} k_j\geq 1$ and $\forall_{j >\ell'}  k_j=0$. 
 Let $N^* = \prod_{j=1}^{\ell'} p_j^{k_j}$ and $N' = \prod_{j=1}^{\ell'} p_j^{k'_j}$. 
 For any $i \geq 2$ and any $\mathbf{x} \in [0,m)^{i-1}$, in the event that $\vec{\beta}_{i-1} = \mathbf{x}$ and $\forall j \ Y_{j, i-1} = k'_j$, then by definition $f_{i-1}(\mathbf{x}) \equiv 0 \bmod N'$ and $\forall_j f_{i-1}(\mathbf{x})/N' \not\equiv 0 \bmod p_j$. In case $i = 1$, we have that $\forall_j Y_{j,0} = 0$, $N' = 1$,  and $\forall_j f_0 = f \neq 0 \bmod p_j$. Thus, for all $i \geq 1$, conditioned on the events that $\vec{\beta}_{i-1} = \mathbf{x}$ and $\forall j \ Y_{j, i-1} = k'_j$, there exist multilinear $\mu-i$ variate polynomials $h_i,g_i$ such that $f_{i-1}(\mathbf{x})/N'=h_i(X_{i+1},\dots,X_\mu)+X_i \cdot g_i(X_{i+1},\dots,X_\mu)$ where for all $j$ at least one of $h_i$ or $g_i$ is nonzero modulo $p_j$. 
 
 % \ben{If we want to extend theorem to higher degree multivariate where the degree of $i$th variable is $d_i$, then in this step we get $f_{i-1}(\mathbf{x})/N' = \sum_{z = 0}^{d_i} h_{z, i} X_i^{z}$ where each $h_{z,i}$ is an $\mu-i$ variate over variables $X_{i+1},...,X_\mu$ preserving the degrees $d_{i+1},...,d_\mu$ on each variable.}
  Furthermore, for any $i\geq 1$, conditioned on the events $\forall j \ Y_{j, i-1} = k'_j$ and $\vec{\beta}_{i-1} = \mathbf{x}$,  the event that $\forall_j Y_{j,i} \geq k_j + k'_j$  %(or in case $i =1$ the event that $\forall_j Y_{j,1} \geq k_j$) 
is equivalent to the event that $h_i + \beta_i g_i \equiv 0 \bmod N^*$.  
 
For each index $t \in [1,2^{\mu-i}]$ let $h_i[t]$ and $g_i[t]$ denote the $t$th coefficients of $h_i$ and $g_i$ respectively (i.e., the coefficients on the $t$th monomial in a canonical ordering of the $2^{\mu-i}$ monomials over the $\mu-i$ variables $X_{i+1},...,X_\mu$). Given that for every $j \in [\ell']$ the polynomial $h_i + X_i \cdot g_i$ is non-zero modulo $p_j$, for each $j \in [\ell']$ there exists at least one index $t_j$ for which the univariate linear polynomial $h_i[t_j] + X_i g_i[t_j]$ is non-zero modulo $p_j$. We now have that: 

%If there is any $j \in [\ell']$ for which $h_i \not\equiv 0 \bmod p_j$ and $g_i \equiv 0 \bmod p_j$ then $h_i + \beta_i \cdot g_i \not\equiv 0 \bmod p_j^{k_j}$ for any value of $\beta_i \in [0,m)$. Otherwise, $g_i \not \equiv 0 \bmod p_j$ for every index $j$, i.e. for all $j \in [\ell']$ there exists an index $t_j$ such that $g_i[t_j]\not\equiv 0 \bmod p_j$. In this second case:

\begin{align}
P[\forall_j Y_{j,i}\geq k_j + k'_j | \vec{\beta}_{i-1} = \mathbf{x} \wedge \forall_j Y_{j,i-1} = k'_j ] 
&=P[\forall_j \ h_i + \beta_i \cdot g_i \equiv 0 \bmod p_j^{k_j}] \\
&\leq P[\forall_j \
h_i[t_j] +\beta_i \cdot  g_i[t_j] \equiv 0  \bmod p_j^{k_j}] \label{eqn:CRT}
\end{align}

%$$=\mathbb{P}_{y \gets [0,m)}[\forall_j f_i(\mathbf{x},y)/N' \equiv 0 \bmod p_j^{k_j}]$$

%$$ =P[\forall_j \ h_i + \beta_i \cdot g_i \equiv 0 \bmod p_j^{k_j}]$$

%$$\leq P[\forall_j \
%h_i[t_j] +\beta_i \cdot  g_i[t_j] \equiv 0  \bmod p_j^{k_j}]$$

%$$= \mathbb{P}_{y \gets [0,m)}[\forall_j \ y \equiv -h_i[t_j] \cdot g_i[t_j]^{-1}  \bmod p_j^{k_j}].$$
For each $t_j$ there is \emph{at most} one solution to the equation $h_i[t_j] + X_i \cdot g_i[t_j] \equiv 0 \bmod \; p_j$. Consequently, by CRT, there is at most one integer solution $x^* \in [0,N^*)$ to the system of equations $\forall_j h_i[t_j] + X_i \cdot g_i[t_j] \equiv \bmod p_j$\footnote{This is where the proof falls apart for higher degree polynomials. For example if the polynomial is quadratic in each variable then there could be up to $2^\ell$ solutions where $\ell$ is the number of prime factors.}. Therefore, if $m \leq N^*$, then the probability in line~(\ref{eqn:CRT}) above is bounded by $1/m$. 

%\ben{In extension to higher degrees, for every $j \in [\ell']$ the polynomial $\sum_{z = 0}^{d_i} h_{z, i} X_i^{z}$ is non-zero modulo $p_j$, which means there exists at least one index $t_j$ where $\sum_{z = 0}^{d_i} h_{z, i}[t_j] X_i^{z} \not\equiv 0 \bmod p_j$. This univariate polynomial has at most $d_i$ zeros. By CRT, there are at most $d_i^{\ell'}$ solutions to the system of equations $\forall_j \sum_{z = 0}^{d_i} h_{z, i}[t_j] X_i^z \equiv 0 \bmod p_j$.}

%Let $y^* \in [0,N^*)$ be the unique value (by CRT) that is equivalent to $-h_i[t_j]\cdot g_i[t_j]^{-1} \bmod p_j^{k_j}$ for all $j\in [\ell']$. 
%\benedikt{Extending this to polynomials with max degree $d$ on each variable we would have that the equation has $d$ solutions per prime. So using the CRT this would result in a $d^\ell$ term. The term factors out over each variable so the final CSZ would be $ \mu\cdot d^\ell \cdot$ the current CSZ.  }

However, we must also consider the case that $m > N^*$. Let $E$ denote the event that the system of equations $\forall_j \
h_i[t_j] +\beta_i \cdot  g_i[t_j] \equiv 0  \bmod p_j^{k_j}$ from line~(\ref{eqn:CRT}) is satisfied. There is at most one integer equivalence class modulo $N^*$ satisfying this system of equations and the random variable $\beta_i$ is uniformly distributed over $[0,m)$ for some $m > N^*$. Let $U$ denote the event that $\beta_i \in [0, m - m \bmod N^*)$ and let $\bar{U}$ denote the event that $\beta_i \in [m - m \bmod N^*, m)$. Conditioned on $U$, $\beta_i$ is uniformly distributed modulo $N^*$, and thus $P[E | U] \leq 1/N^*$. Conditioned on $\bar{U}$, $\beta_i$ is uniformly distributed over the set $[m - m \bmod N^*, m)$. As this set is a contiguous interval of less than $N^*$ integers it contains at most one solution to the systems of equations, and thus, $P[E | \bar{U}] \leq 1/(m \bmod N^*)$. Therefore: 

\begin{align}
P[E] &= P[E | U] \cdot P[U] + P[E | \bar{U}] \cdot P[\bar{U}]\\
&\leq \frac{1}{N^*} \cdot (1-\frac{m \bmod N^*}{m})  + \frac{1}{m \bmod N^*} \cdot \frac{m \bmod N^*}{m}\\
&\leq  
\frac{1}{N^*} +\frac{1}{m}
\end{align}

%Now, for each $j \in[\ell']$ and the random variable $\beta$ sampled from the uniform distribution over $[0,m)$, assuming $g_i[t_j]$ is invertible for each $j \in [\ell]$, let $E_{i,j}$ denote the event that $\beta \equiv -h_i[t_j] \cdot  g_i[t_j]^{-1} \bmod p_j^{k_j}$ so that $\cap_j E_{i,j}$ is the event that $\beta \bmod N^* \equiv \beta^*$. Let $U$ denote the event that $\beta \in [0, m - m \bmod N^*)$ and let $O$ denote the ``overflow" event that $\beta \in [m - m \bmod N^*, m)$. Note that $\beta$ conditioned on $U$ is uniformly distributed modulo $N^*$, while conditioned on $O$, it is only uniform in the set $[m - m \bmod N^*, m)$. As this set is a contiguous interval of less than $N^*$ integers it contains at most one element equivalent to $\beta^*$ modulo $N^*$, precisely when $\beta^* < m \bmod N^*$ as integers. Therefore: 
%\begin{align*}
%P[\cap_j E_{i,j}] &= P[\cap_j E_{i,j} | U] \cdot P[U] + P[\cap_j E_{i,j} | O] \cdot P[O]\\
%&=\mathbb{P}_{\beta \gets [0,N^*)}[
%\beta=\beta^* ] \cdot (1-\frac{m \bmod N^*}{m})  + \frac{1[\beta^* < m \bmod N^*]}{m \bmod N^*} \cdot \frac{m \bmod N^*}{m}\\
%&\leq  \mathbb{P}_{\beta \gets [0,N^*)}[\beta \equiv \beta^* ] +\frac{1}{m} =
%\frac{1}{N^*} +\frac{1}{m}
%\end{align*}

We also have that:
$$P[\forall_{j\in [\ell]} Z_{j,i} \geq k_j] = \prod_{j=1}^{\ell'} P[Z_{j,i} \geq k_j] =  \prod_{j=1}^{\ell'} \left(\frac{1}{p_j}\right)^{k_j} = \frac{1}{N^*} $$

Thus, putting it all together, for any $i \geq 1$ and any $\mathbf{x} \in [0,m)^{i-1}$: 
\begin{align*}
&P[\forall_j Y_{j,i}\geq k_j + k'_j | \vec{\beta}_{i-1} = \mathbf{x} \wedge \forall_j Y_{j,i-1} = k'_j ] 
 \leq \frac{1}{m} + \frac{1}{N^*} = \frac{1}{m} + P[\forall_j Z_{j,i} \geq k_j]
\end{align*}

Since the probability bound is independent of $\mathbf{x}$, this implies:
\begin{align*}
&P[\forall_j Y_{j,i}\geq k_j + k'_j |  \forall_j Y_{j,i-1} = k'_j ] 
 \leq \frac{1}{m} + P[\forall_j Z_{j,i} \geq k_j]
\end{align*}

Consequently, for any $i\geq 2$ and $\vec{k} \in \NN^\ell$ where $\forall_j k_j > 0$:

$$P[\forall_j Y_{j,i}  \geq Y_{j,i-1} + k_j] \leq \max_{\vec{k'}} P[\forall_j Y_{j,i}  \geq k_j + k'_j | \forall_j Y_{j,i-1} = k'_j]\leq \frac{1}{m} + P[\forall_j Z_{j,i} \geq k_j] $$

%\ben{In the extension, $P[\forall_j Y_{j,i}  \geq Y_{j,i-1} + k_j] \leq d_i^{\ell'} \cdot (\frac{1}{m} + P[\forall_j Z_{j,i} \geq k_j])$ } 

Similarly, for $i = 1$ and $\vec{k} \in \NN^\ell$ where $\forall_j k_j > 0$: 

$$P[\forall_j Y_{j,i}  \geq k_j] \leq \frac{1}{m} + P[\forall_j Z_{j,i} \geq k_j]$$
\end{proof}

We now prove the full theorem by induction over $i \in [1,\mu]$. Specifically, we will prove the following inductive hypothesis for $i \in [1,\mu]$: 

$$P[f_i(\vec{\beta_i}) = 0 \bmod N] = P[\forall_j Y_{j,i} \geq r_j] \leq \frac{i}{m} + \prod_{j = 1}^\ell  P[\sum_{k = 1}^i Z_{j,k} \geq r_j]$$

Setting $i = \mu$ this is equivalent to the theorem statement. 

%\ben{Higher degrees: same bound multiplied by $(d_1 \cdots d_\mu)^\ell$}

\textbf{Base Case:}
The base case follows directly from the subclaim for the case $i =1$.

$$ P[\forall_j Y_{j,1}\geq r_j] \leq \frac{1}{m} + P[\forall_j Z_{j,1} \geq r_j] = \frac{1}{m} + \prod_{j = 1}^\ell P[Z_{j,1} \geq r_j]$$
 
%./ \ben{Higher degrees base case is $P[\forall_j Y_{j,1}\geq r_j] \leq d_1^{\ell} \cdot (\frac{1}{m} + P[\forall_j Z_{j,1} \geq r_j])$} 
 
\textbf{Induction Step:}
Assume the inductive hypothesis holds for some $1 \leq i < \mu$, i.e.: 

$$P[\forall_j Y_{j,i} \geq r_j] \leq \frac{i}{m} + \prod_{j = 1}^\ell  P[\sum_{k = 1}^i Z_{j,k} \geq r_j]$$

%\ben{Higher degrees: $$P[\forall_j Y_{j,i} \geq r_j] \leq (d_1\cdots d_i)^\ell \cdot \left( \frac{1}{m} + \prod_{j = 1}^\ell  P[\sum_{k = 1}^i Z_{j,k} \geq r_j] \right)$$ } 

We will show this implies the hypothesis holds for $i+1   $: 
 \begingroup
\allowdisplaybreaks
\begin{align*}
	&P[\forall_j Y_{j,i+1}\geq r_j]= \sum_{\vec{k} \in \NN^\ell} P[\forall_j Y_{j,i+1}-Y_{j,i} \geq r_j-k_j | \forall_j Y_{j,i} = k_j ] \cdot P[\forall_j Y_{j,i}=k_j] \\
&\leq \sum_{\vec{k}} (P[\forall_j Z_{j,i+1} \geq r_j-k_j]+\frac{1}{m}) \cdot P[\forall_j Y_{j,i}=k_j] \quad (\textbf{by subclaim})\\
&= \frac{1}{m}+\sum_{\vec{k}} P[\forall_j Z_{j,i+1} \geq r_j-k_j] \cdot P[\forall_j Y_{j,i}=k_j]&\\
&= \frac{1}{m}+\sum_{\vec{\Delta} \in \ZZ^\ell: \forall_j \Delta_j \leq r_j} P[\forall_j Z_{j,i+1} \geq \Delta_j] \cdot P[\forall_j Y_{j,i}=r_j - \Delta_j] \quad (\textbf{change of variables})\\
&= \frac{1}{m}+\sum_{\vec{\Delta} \in \ZZ^\ell: \forall_j \Delta_j \leq r_j} \sum_{\vec{k} \in \NN^\ell} P[\forall_j Z_{j,i+1} = \Delta_j + k_j] \cdot P[\forall_j Y_{j,i}=r_j - \Delta_j]\\
&=\frac{1}{m}+\sum_{\vec{\Delta'} \in \NN^\ell} P[\forall_j Z_{j,i+1}=\Delta'_j] \cdot \sum_{\vec{k}\in \NN^\ell} P[\forall_j Y_{j,i}= r_j-(\Delta'_j - k_j)] \quad (\textbf{c.o.v.})\\
&=\frac{1}{m}+\sum_{\vec{\Delta'} \in \NN^\ell} P[\forall_j Z_{j,i+1}=\Delta'_j] \cdot P[\forall_j Y_{j,i}\geq r_j-\Delta'_j]\\
&\leq\frac{1}{m}+\sum_{\vec{\Delta'}} P[\forall_j Z_{j,i+1}=\Delta'_j] \cdot (P[\forall_j \sum_{i'=1}^i Z_{j,i' } \geq r_j-\Delta'_j] +\frac{i}{m})  \quad (\textbf{inductive hyp.})\\
&=\frac{i+1}{m}+\sum_{\vec{\Delta'}} P[\forall_j Z_{j,i+1}=\Delta'_j] \cdot P[\forall_j \sum_{k=1}^i Z_{j,k} \geq r_j-\Delta'_j] \\
&=\frac{i+1}{m}+P[\forall_j \sum_{k=1}^i Z_{j,k} \geq r_j-Z_{j,i+1}] \quad(\textbf{independence of variables})\\
&= \frac{i+1}{m}+\prod_{j=1}^\ell P[\sum_{k=1}^{i+1} Z_{j,k}\geq k_j] \quad (\textbf{independence of variables})
\end{align*}

\endgroup

\end{proof}
\section{Bounds on the Regularized Beta Function}

The regularized incomplete beta function is defined for $k, \mu \in \NN$ as:

$$I_\epsilon(k,\mu) = (1-\epsilon)^\mu \sum_{j = k}^\infty \binom{\mu+j -1}{j} \epsilon^j $$

Special values are $I_\epsilon(k, 1) = \epsilon^k$, which matches the tail distribution of a geometric variable with parameter $1-\epsilon$, and $I_\epsilon(1, \mu) = 1 - (1-\epsilon)^\mu$, which is the probability that at least one of $\mu$ geometric variables of parameter $1-\epsilon$ is positive. 

\begin{lemma}
	\label{lem:regbeta1}
 $I_\epsilon(k, \mu) \leq (\epsilon \mu)^k$ for all $\mu, k \in \NN$ and $\epsilon \in (0,1)$ where $\epsilon \mu \leq 1/2$.  
\end{lemma}

\begin{proof} For $k = 0$ the statement holds because $I_\epsilon(0,\mu) = 1$. For $\mu = 1$ we have $I_\epsilon(k,1) = \epsilon^k$. It remains to prove the inequality for $\mu \geq 2$ and $k \geq 1$. 
We will use the following ordinary generating function identity for binomial coefficients: 

$$\sum_{j=0}^\infty \binom{a+j}{a} x^j = \frac{1}{(1-x)^{a+1}}$$

This allows us to write $I_\epsilon(k,\mu)$ as: 

$$I_\epsilon(k,\mu) = (1-\epsilon)^\mu \cdot \left(\sum_{j=0}^\infty \binom{\mu+j-1}{j}\epsilon^j - \sum_{j=0}^{k-1} \binom{\mu+j-1}{j}\epsilon^j \right) $$ 

$$ = (1-\epsilon)^\mu \cdot (\frac{1}{(1-\epsilon)^\mu} - \sum_{j=0}^{k-1} \binom{\mu+j-1}{j}\epsilon^j) = 1 - (1-\epsilon)^\mu \sum_{j=0}^{k-1} \binom{\mu+j-1}{j} \epsilon^j$$

Using the geometric series identity $(\epsilon \mu)^k = 1 - (1-\epsilon \mu) \cdot \sum_{j = 0}^{k-1} (\epsilon \mu)^j$, we obtain:

$$(\epsilon \mu)^k - I_\epsilon(k,\mu) = \sum_{j = 0}^{k-1} (1-\epsilon)^\mu \cdot \binom{\mu+j-1}{j} \epsilon^j - (1-\epsilon \mu) (\epsilon \mu)^j $$

$$= \sum_{j = 0}^{k-1} ((1-\epsilon)^\mu \cdot \binom{\mu+j-1}{j} - (1-\epsilon \mu) \mu^j) \epsilon^j  $$

Let $\Phi_{\epsilon,\mu}(k) = (\epsilon \mu)^k - I_\epsilon(k,\mu)$ so that the goal is to show $\Phi_{\epsilon, \mu}(k) \geq 0$ for $k, \mu \in \NN$ where $\mu \geq 2$ and $\epsilon \mu \leq 1/2$. Observe that:

$$\Phi_{\epsilon,\mu}(1) = \epsilon \mu - 1 + (1-\epsilon)^\mu = (1- \epsilon)^\mu - (1 - \epsilon \mu) \geq 0$$ 

$$\lim_{k \rightarrow \infty} \Phi_{\epsilon, \mu}(k) = 0$$

Thus, it suffices to show that $\Phi_{\epsilon, \mu}(k)$ is non-increasing on the interval $k \in [2, \infty)$ when $\epsilon \mu \leq 1/2$ as this implies that $\Phi_{\epsilon, \mu}(k) \geq 0$ for all $k \geq 1$, $\mu \geq 2$, and $\epsilon \mu \leq 1/2$. Moreover, we can easily show this by showing that for all $\mu,j \geq 2$ and $\epsilon \mu \leq 1/2$: 

$$(1-\epsilon)^\mu \cdot \binom{\mu+j-1}{j} \leq (1- \epsilon \mu) \cdot \mu^j $$ 

Letting $R(\mu,j) = \frac{\binom{\mu + j -1}{j}}{\mu^j}$, observe that:

$$R(\mu, j) = \frac{\binom{\mu + j -1}{j}}{\mu^j} = \frac{\prod_{i = 0}^{j-1} (\mu + i)}{j!\ \mu^j} = \prod_{i = 0}^{j-1} \frac{\mu+i}{\mu \cdot (i+1)} $$ 

Since $\mu+i \leq \mu \cdot (i+1)$ for all $\mu \geq 2$ and $i \geq 0$, it follows that $R(\mu, j) \leq R(\mu, 2) = \frac{1}{2} \cdot (1 + \frac{1}{\mu})$ for $j \geq 2$. Furthermore, for $\epsilon \in (0, 1/\mu)$: 

$$\frac{d}{d\epsilon} \frac{(1 - \epsilon)^\mu}{1 - \epsilon \mu} = \frac{(\mu - 1) \epsilon \mu (1-\epsilon)^{\mu-1}}{(1- \epsilon \mu)^2} \geq 0 $$ 

Thus, for $\epsilon \in (0, \frac{1}{2\mu}]$ and $\mu \geq 2$:  

$$\frac{(1-\epsilon)^\mu}{(1-\epsilon \mu)} \cdot R(\mu,j) \leq \frac{(1-\frac{1}{2\mu})^\mu}{1/2} \cdot R(\mu,2)  \leq \frac{(1 + \frac{1}{\mu})}{\sqrt{e}} \leq e^{\frac{1}{\mu} - {1/2}} \leq 1  $$

This completes the proof. 
\end{proof}

\begin{lemma}
\label{lem:regbeta2}
$I_\epsilon(k,\mu) \leq \epsilon^k \cdot k^\mu$ for $k \geq 2\mu$ and $\epsilon \leq 1/2$. 
\end{lemma}
This is tighter than Bound 1 for larger $k$, i.e. when $k^\mu < \mu^k$.
\begin{proof}
Similar to the analysis in Bound 1, define $\Psi_{\epsilon, \mu}(k) = \epsilon^k k^\mu - I_\epsilon(k,\mu)$ so that:

$$\Psi_{\epsilon, \mu}(k) = \epsilon^k k^\mu - 1 + (1-\epsilon)^\mu \sum_{j=0}^{k-1}\binom{\mu+j-1}{j} \epsilon^j $$

Bound 2 holds iff $\Psi_{\epsilon, \mu}(k) \geq 0$ for all $k \geq 2\mu$ and $\epsilon \leq 1/2$. For $\mu = 1$ we have $I_\epsilon(k, 1) = \epsilon^k$ so 

$$\Psi_{\epsilon, 1}(k) = \epsilon^k \cdot k - \epsilon^k \geq 0$$

Furthermore, $\lim_{k \rightarrow \infty} \Psi_{\epsilon, \mu}(k) = 0$. Thus, it suffices to show that $\Psi_{\epsilon, \mu}$ is non-increasing on the interval $[2\mu,\infty)$ for $\mu \geq 2$ and $\epsilon \leq 1/2$. 

Observe that: 

$$\Psi_{\epsilon, \mu}(k+1) - \Psi_{\epsilon,\mu}(k) = \epsilon^{k+1} (k+1)^\mu - \epsilon^k k^\mu + (1-\epsilon)^\mu \binom{\mu+k-1}{k}\epsilon^k$$

$$= \epsilon^k \left[\epsilon(k+1)^\mu - k^\mu + (1-\epsilon)^\mu \binom{\mu+k-1}{k}\right]$$

Defining: 

$$\Delta_\epsilon(k,\mu) := (1- \epsilon)^\mu \frac{\binom{\mu+k - 1}{k}}{k^\mu} + \epsilon (1 + \frac{1}{k})^\mu $$

$\Psi_{\epsilon, \mu}(k)$ is non-increasing on $k \in [2\mu, \infty]$ iff $\Delta_\epsilon(k,\mu) \leq 1$ for all $k \geq 2\mu$. We will prove this for $\mu \geq 2$ and $\epsilon \leq 1/2$. Using the inequality $(1 + \frac{1}{k})^\mu \leq \sqrt{e}$ for $k \geq 2\mu$: 

$$\Delta_\epsilon(2\mu, \mu) \leq (1- \epsilon)^\mu \frac{\binom{3n - 1}{2\mu}}{(2\mu)^\mu} + \epsilon \sqrt{e}  $$

The right hand side is decreasing as $\mu \rightarrow \infty$ and thus for $\mu \geq 2$ and $\epsilon \leq 1/2$:

$$ \Delta_\epsilon(2\mu, \mu) \leq (1 - \epsilon)^2 \cdot \frac{\binom{5}{4}}{4^2} + \epsilon \sqrt{e} = (1 - \epsilon)^2 \cdot \frac{5}{16} + \epsilon \sqrt{e}\ <  1$$ 

To see why this is less than 1 for $\epsilon \leq 1/2$, note that $\frac{d}{d\epsilon} (1- \epsilon)^2 \cdot c_1 + \epsilon \cdot c_2 = 2c_1 \epsilon + c_2 - 2 c_1$ is positive when $\epsilon \geq 0$ and $c_2 \geq 2 c_1$, and $\sqrt{e} > \frac{5}{8}$. Thus, on the interval $\epsilon \in [0, \frac{1}{2}]$, $(1 - \epsilon)^2 \frac{5}{16} + \epsilon \sqrt{e} \leq \frac{1}{4} \cdot \frac{5}{16} + \frac{1}{2} \sqrt{e} = 0.902...$

Finally, since $\Delta_\epsilon(k,\mu)$ is decreasing as $k \rightarrow \infty$:

$$\forall_{k \geq 2\mu, \mu \geq 2, \epsilon \leq 1/2} \ \Delta_\epsilon(k,\mu) \leq \Delta_\epsilon(2\mu, \mu) < 1 $$
\end{proof}

\begin{corollary}\label{cor:regbeta} 
For any prime $p$ and any positive integer $\mu$
	$$P[\sum_{i=1}^\mu X_i \geq r] = I_{\frac{1}{p}}(r, \mu) \leq \begin{cases}
			\frac{r^\mu}{p^r} \text{ if } r\geq 2\mu\\
		(\frac{\mu}{p})^r \text{ if } p\geq 2\mu \\
		1 \ \text{otherwise} \\ 
	\end{cases}$$, where $X_i$ are independent geometric variables with parameter $(\frac{1}{p})$ and $P[X_i\geq r]=\left(\frac{1}{p}\right)^r$
\end{corollary}

\section{Inverse LCSZ (for cryptographers)}
\cref{thm:csz} (LCSZ) bounds the probability $\mathbb{P}_{\mathbf{x} \leftarrow [0,m)^\mu} [f(\mathbf{x} \equiv 0 \bmod N]$ for given values of $\mu, N,$ and $m$, which has the form $\frac{\mu}{m} + \delta_{N,\mu}$. In the case that $N$ is prime, $\delta_{N, \mu} = \frac{\mu}{N}$, which agrees with the standard Schwartz-Zippel lemma applied to $\mu$-linear polynomials. The term $\delta_{N,\mu}$ for composite $N$, which is dependent on both $\mu$ and the factorization of $N$, has a complicated closed form expression in terms of a product of regularized beta functions.

Motivated by applications to cryptography, this section analyzes the inverse: for a given $\mu, \lambda \in \NN$ what size threshold $t(\lambda, \mu) \in \NN$ is sufficient such that $\delta_{N,\mu} \leq 2^{-\lambda}$ for all $N \geq t(\lambda, \mu)$?  
(Cryptographers often need to know how to choose parameters in order to achieve a target probability bound). In other words: 

\begin{equation}
 \label{eq:optimization}
 	t(\lambda, \mu) :=\sup \{N \in \NN: \prod_{(p, r) \in S(N)}  I_{\frac{1}{p}}(r,\mu) \geq 2^{-\lambda} \} \tag{$t(\lambda, \mu)$ def}
 \end{equation}

For $\mu=1$, since $I_{1/p}(r,1) = \frac{1}{p^r}$ and $\prod_{(p, r) \in S(N)}  I_{\frac{1}{p}}(r,\mu) = \frac{1}{N}$, it is easy to see that $t(\lambda, \mu) = 2^{\lambda}$. For $\mu\geq 2$, the value of $t(\lambda, \mu)$ (or even an upper bound) is not nearly as easy to derive. For the rest of this section we will focus on this $\mu\geq 2$ case. We will analytically derive an upper bound to $t(\lambda, \mu)$, showing that $\log t(\lambda, \mu) \in O(\mu^{2+\epsilon}+\frac{\lambda}{\epsilon})$ for any $\epsilon \geq \log_\mu(2)$.  

\begin{theorem}[Inverse LCSZ]
\label{thm:cryptothm}
For all $\mu \geq 2$, $\epsilon \geq \log_\mu(2)$, and all $N$ such that $$\log N\geq 4 \mu^{2+\epsilon}+ (1+\frac{1}{\epsilon})\cdot \lambda$$
we have that for any $\mu$-linear polynomial $f$ that is coprime with $N$
	$$\mathbb{P}_{x\gets [0,m)^\mu} [f(x)\equiv 0 \bmod N]\leq 2^{-\lambda} +\frac{\mu}{m}$$
	
\end{theorem} 
By setting $\epsilon=\log_\mu(2)$ we get:
\begin{corollary}
	For all $N$ such that $$\log N\geq 8 \mu^{2}+\log_2(2\mu)\cdot \lambda$$
we have that for any $n$-linear polynomial $f$ that is coprime with $N$
	$$P_{x\gets [0,m)^\mu} [f(x)\equiv 0 \bmod N]\leq 2^{-\lambda} +\frac{\mu}{m}$$

\end{corollary}

\subsection{Proof of Inverse LCSZ (\cref{thm:cryptothm})} 
\label{sec:proofinverselcsz}
 By \cref{thm:csz} (CSZ) we have that for $N=\prod_i p_i^{r_i}:$

$$\mathbb{P}_{\xvec\gets [0,m)^\mu}[f(\xvec)\equiv 0 \bmod N]\leq \prod_{i} I_{\frac{1}{p_i}}(r_i,\mu) + \frac{\mu}{m}.$$

For $\mu=1$ and $\log_2(N)\geq \lambda$, \cref{thm:csz} shows that $\mathbb{P}_{x\gets [0,m)} [f(x)\equiv 0 \bmod N]\leq 2^{-\lambda} +\frac{1}{m}$. This is derived by substituting $I_{1/p}(r, 1) = \frac{1}{p^r}$, which gives $\mathbb{P}_{x\gets [0,m)} [f(x)\equiv 0 \bmod N]\leq \frac{1}{N}+\frac{1}{m}$. The case $\mu\geq 2$ is more complicated. This is the focus of the rest of the proof.. 

For a given $N \in \NN$, let $S(N)$ denote the set of pairs $(p, r)$ where $p$ is a prime divisor of $N$ and $r$ is its multiplicity, i.e. $N = \prod_{(p, r) \in S} p^r$. Define: 

 \begin{equation}
 \label{eq:optimization}
 	t(\lambda, \mu) :=\sup \{N \in \NN: \prod_{(p, r) \in S(N)}  I_{\frac{1}{p}}(r,\mu) \geq 2^{-\lambda} \} \tag{$t(\lambda, \mu)$ def}
 \end{equation}
 
 It follows from \cref{thm:csz} (CSZ) that if $N \geq t(\lambda, \mu)$ then 

$$\mathbb{P}_{\xvec\gets [0,m)^\mu}[f(\xvec)\equiv 0 \bmod N]\leq 2^{-\lambda}  + \frac{\mu}{m}.$$

 %Define: 
 
  %\begin{equation}
 %\label{eq:weight}
 %	w_\mu(N) := \sum_{(p,r) \in S(N)} \weight_\mu(p,r)  \tag{$w_\mu(N)$ def}
 %\end{equation}
 
Assuming $t(\lambda, \mu) < \infty$, we obtain the following constrained maximization problem:
\begin{equation}
\label{eq:maximization}
\tag{Constrained Maximization 1}
	\begin{split}
	 \log_2 t(\lambda, \mu) := &\max_{N \in \NN}   \log_2 N \  \text{subject to } \ \sum_{(p,r) \in S(N)} -\log_2 I_{\frac{1}{p}}(r,\mu) \leq  \lambda. 
\end{split}
\end{equation}

In order to derive an upper bound on $\log_2 t(\lambda, \mu)$, we construct a sequence of modified maximization problems, each of which is an upper bound to the prior. The last in this sequence is a knapsack problem for which we analytically derive an upper bound. 

\begin{definition}\label{def:valweight} 
Let $\val(p,r) := r \log_2 p$ and let $\weight_\mu(p,r) = - \log_2 I_{\frac{1}{p}}(r,\mu)$. Additionally, for all $\epsilon\geq \log_\mu(2)$, let: 
	$$\weight_{\mu,\epsilon}(p,r):=\begin{cases}
r \cdot (\log_2(p)-\log_2(\mu)) \text{ if } p\geq \mu^{1+\epsilon}\\
r \cdot \log_2(p)-\mu\cdot \log_2(r) \text{ if } p< \mu^{1+\epsilon} \wedge r > 2(1+\epsilon) \frac{\mu  \ln \mu}{\ln p}  \\
0 \text{ otherwise}
\end{cases}$$
\end{definition}

\begin{claim}\label{claim:weightbound} 
For any prime $p$ and $r \in \NN$,  if $\epsilon \geq \log_\mu(2)$ then $\weight_\mu(p,r) \leq \weight_{\mu,\epsilon}(p,r)$. 
\end{claim} 
\begin{proof} 
If $\epsilon \geq \log_\mu(2)$ then $\mu^{1 + \epsilon} \geq 2\mu$ and the claim follows from Corollary~\ref{cor:regbeta} (to \cref{lem:regbeta1} and \cref{lem:regbeta2}).
\end{proof} 

\begin{claim}  \label{claim:nondecreasingweight} 
For any prime $p$, $\weight_{\mu,\epsilon}(p,r)$ is non-decreasing over $r \geq 1$ and increasing for $r >2 (1+\epsilon) \frac{\mu \log \mu}{\log p}$. 
\end{claim}
\begin{proof} 
We first show that the function $\weight_{\mu,\epsilon}(p,r)$ is non-decreasing in $r$ in each of the three cases, which comprise three subdivisions of the plane $\textsf{Primes} \times \NN$, which we denote $S_A$, $S_B$, and $S_C$ respectively. $S_A$ contains all $(p,r)$ where $p \geq \mu^{1+\epsilon}$, in which case $\frac{d}{dr} \weight_{\mu,\epsilon}(p,r)= \log_2(p) - \log_2(\mu) > 0$. $S_B$ contains all $(p,r)$ where $r > 2(1+\epsilon) \frac{\mu  \ln \mu}{\ln p} $ and $p < \mu^{1+\epsilon}$, in which case $\frac{d}{dr} \weight_{\mu,\epsilon}(p,r) = \log_2(p) -\frac{\mu}{r \ln 2}$, and $\log_2(p) -\frac{\mu}{r \ln 2} \geq \log_2 p - \frac{1}{2 \ln 2} > 0$. The weight function is constant at $0$ for all remaining pairs, which comprise set $S_C$. 

It remains to show that $\weight_{\mu,\epsilon}(p,r)$ increases in $r$ across the boundary between $S_B$ and $S_C$, for which it suffices to show that the weight is positive for all $(p,r) \in S_B$. Suppose, towards contradiction, that $r \log_2 p \leq \mu \log_2 r$ and $r > 2(1+\epsilon) \frac{\mu  \log_2 \mu}{\log_2 p}$. This would imply that both $\frac{r}{\log_2 r} \leq \frac{\mu}{\log_2 p}$ and $\frac{r}{2(1+\epsilon)\log \mu} > \frac{\mu}{\log_2 p}$, which implies that $\log_2 r > 2(1+\epsilon) \log \mu$. Since $\frac{r}{\log_2 r}$ is monotonic increasing in $r$, this in turn implies that $\frac{r}{\log_2 r} > \frac{\mu^2}{4 \log_2 \mu} \geq \mu$ for all $\mu \geq 1$. Finally, the implication that $\frac{r}{\log_2 r} > \mu$ contradicts the assumption that $r \log_2 p  \leq  \mu \log_2 r$. 

\if 0 
Plugging in $r=2(1+\epsilon) \frac{\mu  \log_2 \mu}{\log_2 p}$ to the expression $r\cdot \log_2(p)-\mu \log_2(r)$ we get that
\begin{equation}
	2(1+\epsilon)   \log_2 \mu\geq \log_2(2(1+\epsilon) \frac{\mu  \log_2 \mu}{\log_2 p})
	\end{equation}

The right hand side is monotonically decreasing in $p$ so the inequality is tightest at $p=2$. This gives us:
\begin{equation}
\label{eq:ineq2}
2(1+\epsilon)   \log_2 \mu\geq \log_2(2(1+\epsilon) \mu  \log_2 \mu)
\end{equation}
This has the form $x \log_2(\mu)-\log_2(x \cdot \mu)$ for $x=2\cdot (1+\epsilon)$ which is non decreasing in $\mu$ whenever $x\geq 1$ which is the case as $\epsilon\geq 0$. This implies that the \cref{eq:ineq2} is tightest for $\mu=2$ which gives us: 
\begin{equation}
	2\epsilon  \geq  \log_2(1+\epsilon) 
	\end{equation} 
	This holds for all $\epsilon\geq 0$ which proves that at $r=2(1+\epsilon) \frac{\mu  \ln \mu}{\ln p}$, $\weight_{\mu,\epsilon}$ is non negative.
	This shows that $\weight$ is non decreasing in $r$. 
\fi 
\end{proof} 

The first modified maximization problem is: 

\begin{equation}
\label{eq:maximization2}
\tag{Constrained Maximization 2}
	\begin{split}
	\max_{N \in \NN}   \sum_{(p,r) \in S(N)} \val(p,r) \  \text{subject to } \ \sum_{(p,r) \in S(N)} \weight_{\mu, \epsilon}(p,r) \leq \lambda 
\end{split}
\end{equation}

\begin{claim}\label{claim:hybrid1} 
\cref{eq:maximization2} is an upper bound to \cref{eq:maximization} for any $\mu$ and $\epsilon \geq \log_\mu(2)$: 
\end{claim} 
\begin{proof} 
For any $N \in \NN$, by definition $\sum_{(p,r) \in S(N)} \val(p,r) = \log_2 N$. Thus the only difference between the two maximization problems are the constraints. Furthermore, if $\epsilon \geq \log_\mu(2)$ then, by Claim~\ref{claim:weightbound}, \cref{eq:maximization2} is simply a relaxation of the constraints in \cref{eq:maximization}. 
\end{proof}

Let $\textsf{Primes}$ denote the infinite set of prime numbers.
\begin{definition}[Marginal Value/Weight] \label{def:margvalweight}
 Using $\val(p,r)$ and $\weight(p,r)$ as defined in Definition~\ref{def:valweight}, $r \geq 1$ and $p \in \textsf{Primes}$, let $\margval(p,r) := \val(p, r) - \val(p, r-1) = \log_2 p$ and $\margval(p, 0) := \val(p, 0) = 0$. Similarly, for $r \geq 1$ let $\margweight_{\mu,\epsilon}(p,r) := \weight_{\mu,\epsilon}(p,r) - \weight_{\mu,\epsilon}(p, r-1)$ and $\margweight_{\mu,\epsilon}(p,0) := \weight_{\mu,\epsilon}(p, 0) = 1$. 
 \end{definition} 
 
  By Claim~\ref{claim:nondecreasingweight}, $\margweight_{\mu, \epsilon}(p,r)$ is non-negative for all prime $p$ and $r \in \NN$, and positive for $r >2 (1+\epsilon) \frac{\mu \log \mu}{\log p}$. The following knapsack problem gives an upper bound to \cref{eq:maximization2} : 

\begin{equation}
\label{eq:knapsackproblem}
\tag{Knapsack Problem}
 \max_{S \subseteq  \textsf{Primes} \times \NN } \sum_{(p,r) \in S} \margval(p,r)
\ \text{subject to } \sum_{(p,r) \in S} \margweight_{\mu,\epsilon}(p,r) \leq \lambda. 
\end{equation}

\begin{claim}\label{claim:hybrid2}
 \cref{eq:knapsackproblem} is an upper bound to \cref{eq:maximization2}. 
\end{claim} 
\begin{proof} 
Le $S^*$ denote argmax of \cref{eq:knapsackproblem} with $v^* =  \sum_{(p,r) \in S^*} \margval(p,r)$ and suppose (towards contradiction) that there exists $N \in \NN$ 
such that $\sum_{(p,r) \in S(N)} \weight_{\mu,\epsilon}(p,r) \leq \lambda$ and $\sum_{(p,r) \in S(N)} \val(p,r) > v^*$.  Consider the set $S'$, which includes $S(N)$ and adds for each $(p,r) \in S(N)$ the pairs $(p, j)$ for each $j \leq r$, i.e: 

$$S' = \bigcup_{(p,r) \in S(N)} \bigcup_{j = 0}^r \{(p, j)\}$$

Observe that: 

$$\sum_{(p,r) \in S'} \margval(p,r) = \sum_{(p,r) \in S(N)} \sum_{j = 0}^r \margval(p, j) = \sum_{(p,r) \in S(N)} \val(p,r)   > v^*$$ 
$$ \sum_{(p,r) \in S'} \margweight_{\mu,\epsilon}(p,r) = \sum_{(p,r) \in S(N)} \sum_{j = 0}^r \margweight_{\mu,\epsilon}(p, j) = \sum_{(p,r) \in S(N)} \weight_{\mu,\epsilon}(p,r)  \leq \lambda  $$ 

This is a contradiction to the assumption that $v^*$ is the solution to \cref{eq:knapsackproblem}.

\end{proof}

Finally, we prove a series of claims that will enable us to derive an upper bound on \cref{eq:knapsackproblem}. First, we define: 

\begin{definition}[Density] \label{def:density} 
 $\dens_{\mu, \epsilon}(p,r) = \frac{\margval(p,r)}{\margweight_{\mu,\epsilon}(p,r)}$ where $\margval(p,r)$ and $\margweight(p,r)$ are defined in Definition~\ref{def:margvalweight}.
 \end{definition} 

\if 0 
\begin{claim} \label{claim:maxdensitybound} 
 For all $\mu \geq 1$ and $N \in \NN$: 
 $$\log_2 N \leq w_\mu(N) \cdot \max_{(p,r) \in S(N)} \{ \dens_{\mu,\epsilon}(p,r)\}$$
\end{claim} 
\begin{proof} 
\begin{align*}
w_\mu(N) \cdot \max_{(p,r) \in S(N)} \{ \dens_{\mu,\epsilon}(p,r)\} &\geq \sum_{(p,r) \in S(N)} \weight_\mu(p,r) \cdot \dens_{\mu,\epsilon}(p,r) \\
&= \sum_{(p,r) \in S(N)} \val(p,r) = \log_2 N
\end{align*} 
\end{proof} 
\fi 

\begin{claim} \label{claim:maxdensitybound} 
 For all $S \subseteq \textsf{Primes} \times \NN$: 
 $$\sum_{(p,r) \in S} \margval(p,r)  \leq \sum_{(p,r) \in S} \margweight(p,r) \cdot \max_{(p,r) \in S} \{ \dens_{\mu,\epsilon}(p,r)\}$$
\end{claim} 
\begin{proof} 
\begin{align*}
\sum_{(p,r) \in S} \margweight(p,r) \cdot \max_{(p,r) \in S} \{ \dens_{\mu,\epsilon}(p,r)\} &\geq \sum_{(p,r) \in S} \weight_\mu(p,r) \cdot \dens_{\mu,\epsilon}(p,r) \\
&= \sum_{(p,r) \in S} \val(p,r) 
\end{align*} 
\end{proof}

\begin{claim} \label{claim:densitybound1} 
If $\mu \geq 2$, $\epsilon \geq \log_\mu(2)$, $r \geq 1$, and $p \geq \mu^{1+\epsilon}$ then $\dens_{\mu,\epsilon}(p, r) \leq 1 + \frac{1}{\epsilon}$
\end{claim}

\begin{proof}
If $\epsilon \geq \log_\mu(2)$ and $p \geq \mu^{1+\epsilon}$  then $p \geq 2\mu$, and thus:

$$ \dens_{\mu,\epsilon}(p, r) = \frac{\margval(p,r)}{\margweight_{\mu,\epsilon}(p, r)}  = \frac{\log p}{\log_2(p) - \log_2(\mu)}$$

Furthermore, $p \geq \mu^{1+\epsilon}$ implies that: 
$$\log_2(p) - \log_2(\mu) \geq \log_2(p) - \frac{1}{1+\epsilon} \log_2(p) = \log_2(p) \cdot \frac{\epsilon}{1+\epsilon}$$ 

Thus $\dens_{\mu,\epsilon}(p, r)  \leq \frac{\log_2(p)}{\log_2(p) \cdot \frac{\epsilon}{1+\epsilon}} = 1 + \frac{1}{\epsilon}$
\end{proof}

\begin{claim} \label{claim:densitybound2}
If $\mu \geq 2$, $\epsilon \geq \log_\mu(2)$, $r > 2(1+\epsilon) \cdot \frac{\mu \cdot  \ln \mu}{\ln p}$, and $p < \mu^{1+\epsilon}$ then $\dens_{\mu,\epsilon}(p, r) \leq 1 + \frac{1}{\epsilon}$
\end{claim}
\begin{proof}
Since $p <\mu^{1+\epsilon}$, the conditions on $r$ imply $r > 2(1+\epsilon) \frac{\mu  \ln \mu}{\ln p} > 2\mu$ and thus:

 $$\dens_{\mu,\epsilon}(p, r) = \frac{\margval(p,r)}{\margweight_{\mu,\epsilon}(p, r)} = \frac{\ln p}{\ln p + \mu \ln \frac{r-1}{r}} = \frac{1}{1 - \frac{\mu}{\ln p}\ln\frac{r}{r-1}}$$
 
$\dens_{\mu,\epsilon}(p, r)$ is non-negative because $\margval(p,r)$ is non-negative and $\margweight(p,r)$ is positive over $r >2 (1+\epsilon) \frac{\mu \log \mu}{\log p}$. Thus, for $p$ and $r$ satisfying these conditions, it must be the case that $0 \leq \frac{\mu}{\ln p} \ln \frac{r}{r-1} < 1$. Furthermore, this term is decreasing (approaching zero) as $r$ increases, which shows that for $r$ and $p$ subject to these conditions $\dens_{\mu,\epsilon}(p, r)$ is also decreasing in $r$. 
Combining this with the fact that $\ln \frac{r}{r-1} = \ln (1 + \frac{1}{r-1}) \leq \frac{1}{r-1}$:

$$\dens_{\mu,\epsilon}(p,r) \leq \dens_{\mu, \epsilon}(p, 2(1+\epsilon) \frac{\mu \ln \mu}{\ln p} + 1) \leq \frac{1}{1- \frac{\mu}{\ln p}\frac{\ln p}{2(1+\epsilon)\mu \ln \mu}} = \frac{1}{1-\frac{1}{2(1+\epsilon)\ln \mu}}\leq  1+\frac{1}{\epsilon}$$ 
\end{proof}

\begin{claim} \label{claim:valuebound} 
For $\alpha \in \mathbb{R}$ let $\textsf{Primes}(\alpha)$ denote the set of prime numbers strictly less than $\alpha$.  For $\mu \in \NN$ and $\epsilon \in (0,1)$ define:

$$B_{\mu, \epsilon} := \{(p, r): p \in \textsf{Primes}(\mu^{1+\epsilon}), r \leq 2(1+\epsilon)\frac{\mu \ln \mu}{\ln p} \}$$

Then : 
$$\sum_{(p,r) \in B_{\mu,\epsilon}} \margval(p,r)  \leq 4 \mu ^{2+\epsilon}$$
\end{claim}
\begin{proof} 
Let $\pi(X)$ denote the prime counting function. We use the fact that $\pi(x)\leq 1.3 \cdot \frac{x}{\ln(x)}$ for all $x>1$ \cite{rosser1962approximate}. 
\begin{align*} 
\sum_{(p, r) \in B_{\mu,\epsilon}} \margval(p,r) &= \sum_{p \in \textsf{Primes}(\mu^{1+\epsilon})} \sum_{r = 0}^{\lfloor 2(1+\epsilon) \frac{\mu \ln \mu}{\ln p} \rfloor} \log p  = \sum_{\textsf{Primes}(\mu^{1+\epsilon})} 2(1+\epsilon) \mu \log_2 \mu \\
& \leq 1.3 \cdot \frac{\mu^{1+\epsilon}}{\ln(\mu^{1+\epsilon})}\cdot  \frac{2 \cdot (1+\epsilon) \mu \ln(\mu)}{\ln(2)} \leq 4 \mu^{2+\epsilon} 
 \end{align*} 
\end{proof} 

Putting together these claims, we obtain the bound: 

\begin{claim} \label{claim:analyticalbound} 
For all $\mu \geq 2$, $\epsilon \geq \log_\mu(2)$ and $S \subseteq \textsf{Primes} \times \NN$: 
$$ \sum_{(p,r) \in S} \margval(p,r) \leq 4n^{2 + \epsilon} +  \sum_{(p,r) \in S} \margweight_{\mu,\epsilon}(p,r) \cdot (1 + \frac{1}{\epsilon}) $$
\end{claim} 
\begin{proof} 
Partition $S$ into disjoint sets $S_1$ and $S_2$ such that $S_2$ contains all the pairs $(p, r) \in S$ for which either $p \geq \mu^{1+\epsilon}$ or $r > 2(1+\epsilon) \frac{\mu \ln \mu}{\ln p}$, and $S_1$ contains the remaining pairs. $S_1 \subseteq B_{\mu, \epsilon}$ from Claim~\ref{claim:valuebound} and thus $\sum_{(p,r) \in S_1} \margval(p,r)  \leq 4 \mu ^{2+\epsilon}$. Furthermore, by Claim~\ref{claim:densitybound1}, if $(p,r) \in S_2$ then $\dens_{\mu,\epsilon}(p,r) \leq 1 + \frac{1}{\epsilon}$ and by Claim~\ref{claim:maxdensitybound}: 

$$\sum_{(p,r) \in S_2} \margval(p,r) \leq \sum_{(p,r) \in S_2} \margweight_{\mu, \epsilon}(p,r) \cdot (1 + \frac{1}{\epsilon})$$

Putting everything together: 

$$  \sum_{(p,r) \in S} \margval(p,r) = \sum_{(p,r) \in S_1}  \margval(p,r)  + \sum_{(p,r) \in S_2}  \margval(p,r) \leq 4n^{2 + \epsilon} +  \sum_{(p,r) \in S} \margweight_{\mu,\epsilon}(p,r) \cdot (1 + \frac{1}{\epsilon}) $$
\end{proof} 

Finally, we can conclude from Claim~\ref{claim:analyticalbound} that for any $S \in \textsf{Primes} \times \NN$, $\mu \geq 2$, and $\epsilon \geq \log_\mu(2)$, if $\sum_{(p,r) \in S} \margweight_{\mu,\epsilon}(p,r) \leq \lambda$, i.e., if $S$ satisfies the constraints of \cref{eq:knapsackproblem} then: 

$$\sum_{(p,r) \in S} \margweight_{\mu,\epsilon}(p,r) \leq 4n^{2 + \epsilon} + \lambda \cdot (1 + \frac{1}{\epsilon}) $$

The right hand side of the equation is therefore an upper bound for the solution to \cref{eq:knapsackproblem}, and consequently (by Claim~\ref{claim:hybrid1} and Claim~\ref{claim:hybrid2}), an upper bound for the solution $t(\lambda,\mu)$ to \cref{eq:maximization} when $\mu \geq 2$. In conclusion, for any $N \in \NN$, $\mu \geq 2$, and $\epsilon \geq \log_\mu(2)$, if $\log_2 N \geq 4n^{2 + \epsilon} + \lambda \cdot (1 + \frac{1}{\epsilon})$ then $\log_2 N \geq t(\lambda, \mu)$ and: 

$$\mathbb{P}_{\xvec\gets [0,m)^\mu}[f(\xvec)\equiv 0 \bmod N]\leq 2^{-\lambda}  + \frac{\mu}{m}$$

\subsection{Computational Inverse LCSZ} 
 \cref{thm:cryptothm} provides an analytical upper bound on $t(\lambda, \mu)$ for any $\mu, \lambda \in \NN$. However, the analytical bound does not appear to be tight for $\mu \geq 2$ (for $\mu =1$ it is tight). This next section provides an algorithm to derive an upper bound on $t(\lambda, \mu)$ for any specific values of $\lambda, \mu$. The algorithm gives tighter bounds than \cref{thm:cryptothm} for a table of tested values (Table~\ref{tab:expresults}).  This is useful in practice, e.g. for deriving concrete cryptographic security parameters in cryptographic protocols that rely on LCSZ.

As shown in the prior section, $t(\lambda, \mu)$ is upper bounded by a solution to a knapsack problem, \cref{eq:knapsackproblem}, over the infinite set of items $\textsf{Primes}\times \NN$. There is a simple well-known greedy algorithm that returns an upper bound to the optimal value for the knapsack problem over a \textit{finite} set of items. This algorithm greedily adds items to the knapsack in order of decreasing density until the knapsack overflows the weight bound, and returns the total value of the added items at this point. Over an infinite set of items, it is not generally possible to sort by decreasing density. However, by leveraging monotonicity properties of the density function in our case, we are able to adapt the greedy approximation algorithm to work for \cref{eq:knapsackproblem}. In particular, we are able to enumerate over pairs in $\textsf{Primes}\times \NN$ in an order of non-increasing density. 

\begin{claim} \label{claim:enumeration} 
Let $A$ and $B$ be any pair of discrete strictly ordered sets which contain minimum elements $a_0 \in A$ and $b_0 \in B$. Let $a+$ denote the next element of $A$ after $a$ and likewise $b+$ the next element of $B$ after $b$. If $f: A\times B \rightarrow \mathbb{R}^+$ is monotonically non-increasing over pairs $(a, b_0)$ as $a \in A$ increases, and for any fixed $a$, monotonically non-increasing over pairs $(a, b)$ as $b \in B$ increases, then the following algorithm enumerates the pairs $(a,b) \in A \times B$ in order of decreasing $f(a,b)$. The algorithm initializes the set $C = \{(a_0, b_0)\}$ and also a variable $\textsf{max}_A$ to keep track of the highest order element in $A$ seen so far.  At each iteration, it removes a pair $(a, b) \in C$ of lowest $f(a, b)$ value and appends $(a,b)$ to the output enumeration list. Before proceeding to the next iteration, it adds $(a, b+)$ to $C$, and if $a = \textsf{max}_A$ then it also adds $(a+, b_0)$ to $C$ and updates $\textsf{max}_A := a+$.
\end{claim} 
\begin{proof} 
Suppose, towards contradiction, that the algorithm appends $(a, b)$ to the output list, and there exists at least one pair $(a', b')$ not yet in the list at this iteration such that $f(a', b') > f(a, b)$. If $b' \neq b_0$ and $(a', b_0)$ appeared in the output list already, then each $(a', b^*)$ for $b^* \in [b_0, b']$ would also have been added to $C$ and removed before $(a, b)$ because $f(a', b^*)\geq f(a', b') > f(a, b)$. This would be a contradiction, so it remains to consider the case that $b' = b_0$, i..e. that $(a', b_0)$ did not appear in the list and $f(a', b_0) > f(a, b)$. 

First, this implies that $a' > a$. Otherwise, If $a' \leq a$, then $(a', b_0)$ would have been added to $C$ before $(a, b)$ and thus removed before $(a, b)$. 
Second, $b \neq b_0$ and $f(a, b) < f(a, b_0)$, as otherwise this would imply that $f(a', b_0) > f(a, b) \geq f(a, b_0)$, contradicting the monotonicity property. Thus $(a, b_0)$ must already appear in the output list because it is added to $C$ before $(a, b)$ and removed before $(a, b)$. Furthermore, for all $a^* \in [a, a']$, $f(a^*, b_0) \geq f(a', b_0)  > f(a, b)$, thus each such pair $(a^*, b_0)$ would have been added and removed before $(a, b)$. This is a contradiction. 
\end{proof}

We use the enumeration algorithm of Claim~\ref{claim:enumeration} to implement the greedy algorithm that obtains an upper bound to a generic knapsack problem over the infinite set of items $\textsf{Primes}\times \NN$: 

\begin{equation}
\label{eq:genericknapsackproblem}
\tag{Generic Knapsack Problem}
 \max_{S \subseteq  \textsf{Primes} \times \NN } \sum_{(p,r) \in S} \val(p,r)
\ \text{subject to } \sum_{(p,r) \in S} \weight(p,r) \leq \lambda. 
\end{equation}
 
for any $\val, \weight: \textsf{Primes}\times \NN \rightarrow \mathbb{R}^+$ 
where $\dens(p, r) = \frac{\val(p,r)}{\weight(p,r)}$ is monotonic non-increasing over $r$ for any fixed $p$, and also over $p$ for fixed $r = 1$.  This is presented below as \cref{alg:greedyalg2}. 

\begin{algorithm}[hbt]
Input $\mu \in \NN ,\lambda \in \NN$

\begin{enumerate}  
%\item Input: $\mu \in \ZZ, \lambda \in \ZZ^+$
%\item Set $p_max =2$
%\item Set $N=1$
\item Initialize a max heap $H$ that stores tuples $(p,r,d) \in \PP\times \ZZ \times \RR$ and sorts them by the third value.
 \item Initialize $w\gets 0$ and $v \gets 0$.
\item Push $(\dens(2,1),2,1)$ onto the heap and set $\textsf{pmax} = 2$. 
\item While $w<\lambda$
\item \begin{enumerate}
\item Pop $(p,r,d)$ from $H$.
\item Push $(p,r+1, (\dens(p,{r+1}))$ onto $H$
\item Set $v\gets v+ \val(p,r)$
\item Set $w \gets w +\weight(p,r)  $
%\item Set $N\gets N\cdot p$
\item If $p=\textsf{pmax}$ then set $\textsf{pmax} \gets \mathsf{next\_prime}(p)$ and push $(\dens(\textsf{pmax},1),\textsf{pmax},1)$ onto $H$
\end{enumerate}
\item Output $v$ %N$
\end{enumerate}

\caption{Greedy algorithm that returns an upper bound to \cref{eq:maximization2} }
\label{alg:greedyalg2}

\end{algorithm}

Moreover, we will show that the density function defined in terms of $\val(p, r) = \log p$, $\weight(p, 1) = -\log I_{1/p}(1,\mu)$, and $\weight(p,r) = \log I_{1/p}(r-1,) - \log I_{1/p}(r,\mu)$ for $r > 1$ satisfies this monotonicity property. The density function in \cref{eq:knapsackproblem} also has this monotonicity property, but we are able to obtain a tighter bound on \cref{eq:maximization} by defining regularized beta function directly instead of the simpler form upper bounds on the regularized beta function that were more useful for deriving the analytical result in \cref{thm:cryptothm}. 

\begin{claim}\label{claim:hybrid1}
 \cref{eq:genericknapsackproblem} with $\dens(p,r) =
\frac{\val(p,r)}{\weight(p,r)}  $, $\val(p, r) = \log p$, and $\weight(p,r) = \log I_{1/p}(r-1,\mu) - \log I_{1/p}(r,\mu)$ is an upper bound to \cref{eq:maximization}. 
\end{claim} 
\begin{proof} The analysis is the same as in Claim~\ref{claim:hybrid2}, replacing $\margweight_{\mu, \epsilon}$ with the weights defined here.  
\end{proof} 
\begin{claim} For $p\in \PP, r \in \NN$ and $\mu\geq 2 \in \NN$ let $\dens(p,r) =
\frac{\val(p,r)}{\weight(p,r)}$ where $\val(p, r) = \log p$, and $\weight(p,r) = \log I_{1/p}(r-1,\mu) - \log I_{1/p}(r,\mu)$ for $r \geq 1$. Then $\dens(p,r)$ is decreasing in $r$ and $\dens(p, 1)$ is non-increasing in $p$.
\label{claim:nonincreasingdens} 
\end{claim}

\begin{proof} 

\noindent\textbf{Part 1: $\dens(p,r)$ is decreasing in $r$} 

$\dens(p,r) = \frac{\log p}{\weight(p,r)}$ is decreasing in $r$ iff $\weight(p,r)$ is increasing in $r$. 

$$\weight(p,r) = - \log I_{1/p}(r,\mu) + \log I_{1/p}(r-1,\mu) = \log \frac{I_{1/p}(r-1,\mu)}{I_{1/p}(r,\mu)} $$

is decreasing in $r$ if $\frac{I_{1/p}(r-1,\mu)}{I_{1/p}(r,\mu)}$ is decreasing in $r$. This is the case if for all $r$ 
$$\frac{I_{1/p}(r,\mu)}{I_{1/p}(r+1,\mu)}-\frac{I_{1/p}(r-1,\mu)}{I_{1/p}(r,\mu)}>0$$. 
Which is equivalent to showing that 
\begin{equation}\label{eq:monotonecond}
	I_{1/p}(r,\mu)\cdot I_{1/p}(r,\mu)-I_{1/p}(r-1,\mu)\cdot I_{1/p}(r-1,\mu)>0
\end{equation}
The regularized beta function $I_{x}(a,b)$ is log convave for all $b>1$ as shown in \cite{KarpLogConcave}. This implies that for $b>1$ and all $\alpha,\beta>0$ $I_{x}(a+\alpha,b)\cdot I_{x}(a+\beta,b)-I_{x}(a,b)\cdot I_{x}(a+\alpha+\beta,b)>0$. Setting $a=r-1,b=\mu,\alpha=1,\beta=1$ this shows that \cref{eq:monotonecond} holds and $\weight$ is incresing in $r$, and $\dens$ is decreasing in $r$ for all $\mu>1$. 
 \medskip 

\noindent \textbf{Part 2: $\dens(p,1)$ is non-increasing in $p$.}

$$\dens(p,1) = \frac{\log p}{\weight(p,1)} = \frac{\ln p}{-\ln(1 - (1-p^{-1})^\mu)}$$

The derivative $\frac{d}{dp} \dens(p,1)$ is non-negative iff: 

$$-p^{-1} \cdot \ln(1- (1-p^{-1})^\mu) - \frac{\ln p \cdot \mu (1-p^{-1})^{\mu-1} \cdot p^{-2} }{1- (1-p^{-1})^\mu} \geq 0  $$ 

Equivalently: 

$$ - \ln (1- (1- p^{-1})^\mu) -\frac{\ln p \cdot \mu (1- p^{-1})^{\mu-1} p^{-1}}{1 - (1-p^{-1})^\mu} \geq 0 $$ 

Set $x = 1 - \frac{1}{p}$, which increases over the range $(1/2, 1)$ as $p$ increases in the range $[2,\infty)$. The requisite inequality for $x \in (1/2, 1)$ becomes: 

$$- \ln( 1 - x^\mu) + \frac{\ln (1-x) \cdot \mu \cdot x^{\mu-1} (1-x)}{1-x^\mu}  \geq 0$$ 

which, rearranging terms, holds true iff for $x \in (1/2,1)$: 

$$ \frac{\ln(1-x^\mu)(1-x^\mu)}{\ln(1-x)(1-x)x^{\mu-1}}  \leq   \mu  $$

In fact we can show this inequality holds true over all $x \in (0,1)$. Using the inequalities $\ln(1-x^\mu) \leq -x^\mu$ and $-\ln(1-x) \geq x$ we obtain: 

$$\frac{\ln(1-x^\mu)(1-x^\mu)}{\ln(1-x)(1-x)x^{\mu-1}} \leq \frac{x^\mu (1-x^\mu)}{x(1-x)x^{\mu-1}} = \frac{1-x^\mu}{1-x} \leq \mu $$

\end{proof} 
\begin{theorem}[Computational bound]
	Let $k$ be the output of algorithm \cref{alg:greedyalg2} on input $\lambda,\mu$. Then for all $m\in \NN$, all $N\geq 2^k$ and all $\mu$-linear polynomials $f$, coprime with $N$, $\log_2 N \geq t(\lambda,\mu)$ and 
	$$P_{\xvec \gets [0,m)^\mu}[f(\xvec) \equiv 0 \bmod N]\leq 2^{-\lambda}+\frac{\mu}{m}$$
	\end{theorem}
	\begin{proof}
		\cref{alg:greedyalg2} is a greedy enumeration algorithm over pairs $(p,r)$ following the enumeration strategy in \cref{claim:enumeration}. By \cref{claim:nonincreasingdens}, $\dens$ satisfies the monotonicity conditions required for \cref{claim:enumeration} and thus the enumeration algorithm enumerates pairs in order of non-increasing density. Thus, the algorithm outputs an upper bound to \cref{eq:genericknapsackproblem} with density $\dens(p,r) =\frac{\val(p,r)}{\weight(p,r)} $, $\val(p, r) = \log p$, and $\weight(p,r) = \log I_{1/p}(r-1,\mu) - \log I_{1/p}(r,\mu)$. By \cref{claim:hybrid1} this is an upper bound to \cref{eq:maximization} which in turn by \cref{thm:csz} gives a bound on $\log_2 t(\lambda,\mu)$.
	\end{proof}

\subsection{Computational Results}
Using \cref{alg:greedyalg2} we computed analytical bounds for all $\mu\in(1,50)$ for different values of $\lambda$. 
The precise bound for $\mu=20$ and $\lambda=120$ is 
\begin{dmath*}2^v=2^{36} \cdot 3^{20} \cdot 5^{11} \cdot 7^{8} \cdot 11^{5} \cdot 13^{5} \cdot 17^{4} \cdot 19^{3} \cdot 23^{3} \cdot 29^{2} \cdot 31^{2} \cdot 37^{2} \cdot 41^{2} \cdot 43^{2} \cdot 47^{2} \cdot 53^{2} \cdot 59 \cdot 61 \cdot 67 \cdot 71 \cdot 73 \cdot 79 \cdot 83 \cdot 89 \cdot 97 \cdot 101 \cdot 103 \cdot 107 \cdot 109 \cdot 113 \cdot 127 \cdot 131 \cdot 137 \cdot 139 \cdot 149 \cdot 151 \cdot 157 \cdot 163\end{dmath*}
Other results are presented in their logarithmic form in \cref{tab:expresults}. The results are significantly tighter than the analytical \cref{thm:cryptothm}. For $n=20$ and $\lambda=120$ the analytical theorem gives a value for $\log_2(N)$ of $3839$ vs the computational which is $416$.  
We also provide the open-source Python implementation of the algorithm on Github\footnote{\url{https://github.com/bbuenz/Composite-Schwartz-Zippel}}.  
\begin{table}[h!t]

\begin{mdframed}

\centering
\begin{tabular}{l||l|l|l|l}

$\mu$ & $\lambda=40$ & $\lambda=100$ & $\lambda=120$ & $\lambda=240$\\
\hline
\hline
1 & 40 & 100 & 120 & 240\\
2 & 57 & 130 & 156 & 290\\
3 & 67 & 148 & 175 & 328\\
4 & 79 & 169 & 197 & 359\\
5 & 86 & 187 & 212 & 386\\
6 & 97 & 200 & 234 & 415\\
7 & 107 & 214 & 244 & 435\\
8 & 113 & 227 & 260 & 459\\
9 & 122 & 237 & 277 & 483\\
10 & 133 & 252 & 289 & 500\\
11 & 139 & 263 & 301 & 523\\
12 & 148 & 276 & 315 & 540\\
13 & 152 & 291 & 331 & 565\\
14 & 160 & 304 & 344 & 576\\
15 & 168 & 314 & 354 & 600\\
16 & 178 & 323 & 366 & 616\\
17 & 186 & 335 & 381 & 634\\
18 & 193 & 347 & 391 & 653\\
19 & 198 & 356 & 407 & 664\\
20 & 207 & 368 & 416 & 679\\
21 & 216 & 378 & 429 & 695\\
22 & 222 & 389 & 437 & 718\\
23 & 228 & 402 & 448 & 732\\
24 & 233 & 411 & 464 & 749\\
25 & 241 & 420 & 472 & 758\\
26 & 248 & 432 & 481 & 772\\
27 & 256 & 438 & 492 & 792\\
28 & 264 & 452 & 506 & 806\\
29 & 275 & 460 & 516 & 820\\
30 & 278 & 469 & 527 & 831\\
\hline
50 & 419 & 662 & 736 & 1105
\end{tabular}	
\end{mdframed}

\caption{Computationally determined values of $t(\lambda, \mu)$ such that for all $N\geq t(\lambda, \mu)$, $\mathbb{P}_{\mathbf{x} \gets [0,m)^\mu} [f(\mathbf{x})\equiv 0 \bmod N]\leq 2^{-\lambda}+\frac{\mu}{m}$ for different $\mu$ and different $\lambda$}
\label{tab:expresults}

\end{table}

\pagebreak
\printbibliography

\end{document}